\documentclass{article}

\usepackage{amssymb}

\newtheorem{thm}{Theorem}

\newtheorem{lemma}[thm]{Lemma}

\newtheorem{proposition}[thm]{Proposition}

\newenvironment{proof}{\noindent\bf{Proof.}\rm}{\hfill$\blacksquare$\bigskip}

\begin{document}

\title{Faster FAST
\\ (Feedback Arc Set in Tournaments)}

\author{Uriel Feige~\thanks{Department of Computer Science and Applied Mathematics, the Weizmann Institute, Rehovot, Israel.
{\tt uriel.feige@weizmann.ac.il}. Work supported in part by The
Israel Science Foundation (grant No. 873/08).}}

\maketitle

\begin{abstract}
We present an algorithm that finds a feedback arc set of size $k$
in a tournament in time $n^{O(1)}2^{O(\sqrt{k})}$. This is
asymptotically faster than the running time of previously known
algorithms for this problem.
\end{abstract}

\section{Introduction}

A {\em tournament} is a directed graph in which every pair of
vertices is connected by exactly one arc. A {\em feedback arc set}
is a set of arcs whose removal makes the remaining digraph
acyclic. Given a tournament, the Feedback Arc Set in Tournaments
(FAST) problem asks for the smallest feedback arc set in the
tournament. This problem is NP-hard~\cite{ACN,Alon}. Hence we
shall consider a parameterized version of the problem, $k$-FAST,
in which one is given a tournament and a parameter $k$, and one
has to find a feedback arc set of size $k$ if one exists.
In~\cite{ALS} in was shown (among other things) that this problem
can be solved in time $n^{O(1)} + k^{O(\sqrt{k})}$. (Here and
elsewhere $n$ denotes the number of vertices in the tournament.)
The interesting aspect of this running time is the subexponential
dependence on $k$, as the fact that the problem is fixed parameter
tractable and moreover has a polynomial kernel was established
earlier~\cite{DGHNT}. Given that~\cite{ALS} is titled {\em Fast
FAST}, there is a big temptation to publish a paper titled {\em
Faster FAST}. Not being able to resist this temptation, we present
here a different algorithm that offers a mild improvement to the
running time.

\begin{thm}
\label{thm:main} There is an algorithm that solves $k$-FAST in
time $2^{O(\sqrt{k})} n^{O(1)}$.
\end{thm}

Observe that equivalently, this running time can be written as
$n^{O(1)} + 2^{O(\sqrt{k})}$ (this only changes the constants in
the $O$ notation). The running time of the algorithm of~\cite{ALS}
remains polynomial in $n$ for $k = O((\log n/ \log \log n)^2)$,
whereas the running time of our algorithm remains polynomial for
$k = O((\log n)^2)$.

The algorithm presented in~\cite{ALS} is based on the color coding
technique (introduced in~\cite{AYZ}), and specifically on a
certain random coloring lemma: for every graph with $k$ edges, if
one colors its vertices at random by $O(\sqrt{k})$ colors, then
with probability at least $2^{-O(\sqrt{k})}$ the coloring is
proper (see~\cite{ALS} for an exact statement of this lemma).
In~\cite{ALS} this lemma is used in combination with dynamic
programming to design an algorithm for $k$-FAST. Moreover, this
lemma may be of interest beyond the specific application to the
$k$-FAST problem. The algorithm presented in the current paper is
also based on dynamic programming. However, it does not use the
random coloring lemma.

\section{The algorithm}

The feedback arc set problem is equivalent to finding a linear
ordering of the vertices (numbering them from~0 to $n-1$) that
minimizes the number of arcs pointing backwards. Had the
tournament been acyclic, there would have been a simple local test
that would tell us where to place a vertex $v$ in this linear
ordering. We call it the {\em indegree test}. Under this test, the
proposed location for vertex $v$ is $i$ if and only if $v$ has $i$
incoming arcs (and $n - i - 1$ outgoing arcs).

What happens if the minimum feedback arc set is of size $k > 0$?
In this case, the degree test might be incorrect. Let $\pi$ be an
optimal linear ordering (one with only $k$ backward arcs). Let the
error of the indegree test for vertex $v$ be the absolute value of
the difference between its location under $\pi$ and the number of
incoming arcs that $v$ has. Then the sum of errors over all
vertices is at most $2k$ (since each feedback arc contributes at
most~2 to the error). It follows that for every value of $d$ (we
shall later choose $d = \Theta(\sqrt{k})$) there are at most
$2k/d$ vertices for which the error is more than $d$. Let
$D_{\pi}$ denote the set of all vertices of error more than $d$
with respect to the optimal linear ordering $\pi$. As we shall see
shortly, given $D_{\pi}$, a minimum feedback arc set can be
computed in time $n^{O(1)}2^{O(|D_{\pi}| + d)}$. For $d =
\Theta(\sqrt{k})$ the running time becomes
$n^{O(1)}2^{O(\sqrt{k})}$ as desired. The difficulty is that
$D_{\pi}$ is not given to us, and handling this issue is the
purpose of the following discussion.

We say that three vertices in the tournament form a {\em triangle}
if their corresponding arcs form a directed cycle. At least one of
the arcs in a triangle is a feedback arc. It is known that a
tournament is acyclic if and only if it does not contain
triangles. Call an arc {\em suspect} if it belongs to some
triangle. Call an arc a {\em major suspect} if it belongs to at
least $t$ triangles, where $t = \Theta(\sqrt{k})$ is a parameter
to be chosen later. Call a vertex {\em bad} if at least $t$ arcs
incident with it are major suspects. Let $B$ denote the set of all
bad vertices in a tournament. Clearly, given the tournament, the
set $B$ can be computed in polynomial time. We shall now show that
for appropriate choices of $d$ and $t$, $D_{\pi} \subset B$, and
moreover, that like $D_{\pi}$, the size of $B$ is $O(\sqrt{k})$.

\begin{lemma}
\label{lem:badgood} For $D_{\pi}$ and $B$ as defined above, if $d
\ge 4t$ and $t \ge \sqrt{k}$, then $D_{\pi} \subset B$.
\end{lemma}

\begin{proof}
We need to show that every vertex with error $d' > d$ is incident
with at least $t$ arcs that are major suspects. Consider such a
vertex $v$, let $i$ be its location in $\pi$ and let $i + d'$ be
the number of $v$'s incoming arcs. (The case in which $i - d'$ is
the number of incoming arcs is handled in a similar way and is
omitted.) Let $F$ be the set of vertices that come after $v$ in
$\pi$ and yet have an arc directed from them to $v$. Clearly, $|F|
\ge d'$. Let us denote the vertices in $F$ by $v_1, v_2, \dots$ in
the order of their appearance after $v$. A crucial observation is
that for every $j$, the location of $v_j$ in $\pi$ has to be no
sooner than $i + 2j$. Otherwise, $\pi$ is not an optimal linear
arrangement, because the size of the feedback arc set can be
decreased by doing one cyclic shift on the block of vertices that
starts at $v_i$ and ends at $v_j$ (where all vertices in the block
move one location down except $v_i$ that moves to the original
location of $v_j$). Doing this cyclic shift, $j$ arcs are removed
from the feedback arc set and less than $j$ arcs join the feedback
arc set (only arcs incident with $v_i$ are affected by the cyclic
shift).

Now consider the arc $(v_j,v_i)$, which is a feedback arc in
$\pi$. Let us consider only triples of vertices $(v_i,v_j,u)$
where vertex $u$ has to lie in $\pi$ between $v_i$ and $v_j$, and
moreover, the arc $(v_i,u)$ is directed towards $u$ (in agreement
with the linear order $\pi$). By the observation above, there are
at least $j$ possibilities for the choice of $u$. The triple
$(v_i,v_j,u)$ forms a triangle unless the arc $(v_j,u)$ is
directed towards $u$, which makes it too a feedback arc in $\pi$.

Now we use the fact that the feedback arc set of $\pi$ has size
$k$. Consider only values of $j$ between $2t$ and $4t$ (here we
used the assumptions the $d' \ge d \ge 4t$). For each such vertex
$v_j$, the arc $(v_j,v_i)$ is not a major suspect only if at least
$t$ arcs $(v_j,u)$ are feedback arcs in $\pi$. Hence for $v_i$ not
to be incident with $t$ major suspects, there must be more than
$t^2$ feedback arcs in $\pi$, which is a contradiction for $t \ge
\sqrt{k}$.
\end{proof}

\begin{proposition}
\label{pro:badsmall} For $B$ as defined above and $t = \sqrt{6k}$,
$|B| < t$.
\end{proposition}

\begin{proof}
Each vertex of $B$ is incident with at least $t$ arcs which are
major suspects (we shall think of this as being exactly $t$, by
possibly ignoring some of the major suspects), and each such arc
is in at least $t$ triangles (again, we shall think of this as
being exactly $t$, by possibly ignoring some of the triangles).
This gives a count of $t^2|B|$ triangles (though possibly the same
triangle might be counted more than once). Each such triangle has
at least one feedback arc, and hence we have counted $t^2|B|$
feedback arcs. The problem is that the same arc might have been
counted several times. We bound now the number of times that a
single arc $(u,v)$ might be counted.

The arc $(u,v)$ may serve as a major suspect twice, once for $u$
and once for $v$. This gives $2t$ triangles in which $(u,v)$ might
have been counted. Also for every other major suspect incident
with $u$ (or with $v$), the arc $(u,v)$ might appear in one of its
triangles. This gives another $2t$ triangles in which $(u,v)$
might appear. Finally, for every vertex $w \not\in \{u,v\}$ in
$B$, if either $(w,u)$ or $(w,v)$ is a major suspect, then the arc
$(u,v)$ might again be counted in a triangle. This gives at most
$2|B|$ additional triangles. Altogether an arc might be counted at
most $4t + 2|B|$ times.

As the total number of feedback arcs is $k$, we obtain the
inequality $t^2|B|/(4t + 2|B|) \le k$, which implies the
proposition.
\end{proof}

We can now describe our algorithm. We assume without loss of
generality that the value of $k$ is known (the algorithm may try
all values of $k$ in increasing order until the first one that
succeeds). Given a value of $k$, let us fix $t = 3\sqrt{k}$ and $d
= 12\sqrt{k}$. The main steps of the algorithm are as
follows.

\begin{enumerate}

\item Compute the set $B$ of bad vertices.

\item For each location $i$ in the linear order, compute a
candidate set $C(i)$ that contains those vertices whose indegree
is between $i - d$ and $i + d$, plus the vertices of $B$. In
addition compute a prefix set $P(i)$ that contains those vertices
not in $B$ with indegree less than $i - d$.

\item Using these candidate sets and prefix sets, compute a minimum
feedback arc set using dynamic programming.

\end{enumerate}

We now elaborate on these main steps, proving the correctness of
the algorithm and bounding its running time.

Step (1) can be done in time $O(n^3)$ by checking for each triple
of vertices whether it forms a triangle, then identifying those
arcs that are major suspects (members of at least $t$ triangles),
and putting in $B$ those vertices that are incident with at least
$t$ major suspects. By Proposition~\ref{pro:badsmall} we have that
$|B| \le t$ and by Lemma~\ref{lem:badgood} we have that $D_{\pi}
\subset B$.

Given $B$, Step (2) can also be performed in polynomial time,
since the indegree of vertices can be computed in polynomial time.
The properties required from Step~(2) are summarized in the
following proposition.

\begin{proposition}
\label{pro:sets} In the optimal linear arrangement $\pi$, for
every location $i$, the vertex in location $i$ is one of the
vertices of the candidate set $C(i)$ (as computed in Step~(2) of
the algorithm). Moreover, every vertex of the prefix set $P(i)$ is
placed  in $\pi$ prior to location $i$.
\end{proposition}

\begin{proof}
Let $v$ be the vertex placed by $\pi$ in location $i$. We need to
show that $v \in C(i)$. If $v \in D_{\pi}$ then this follows from
Lemma~\ref{lem:badgood} because $D_{\pi} \subset B$. Hence it
remains to consider the case that $v \not\in D_{\pi}$. In this
case the indegree of $v$ is in the range $[i - d, i+d]$, again
implying that $v \in B$, as desired. A similar argument shows that
all vertices of the prefix set $P(i)$ are placed in $\pi$ before
location $i$.
\end{proof}

Now we can use dynamic programming to find which linear order
among those that respect the candidate sets has the smallest
feedback arc set. We scan the locations from~0 to~$n-1$. On
reaching location $i$ we need only know two things:

\begin{enumerate}

\item Which vertices of $C(i)$ have been placed up to location $i$.
There are at most $2^{|C(i)|}$ possibilities for such subsets.

\item How many backward arcs we have placed so far. For each
choice of subset $C'(i)$ as in item (1), we need to remember just
the smallest number of backward arcs that can attained in a linear
arrangement that up to $i$ placed $C'(i)$ and did not place $C(i)
- C'(i)$.

\end{enumerate}

At step $i$, one can place at location $i$ any one of the vertices
$v$ of $C(i) - C'(i)$. (If $C(i) - C'(i)$ is empty, the
corresponding branch of the dynamic programming dies off.)
Thereafter, $C'(i+1)$ can be computed in a straightforward way as
$C(i+1) \cap (C'(i) \cup \{v\})$. Likewise, the number of backward
arcs can be updated by adding to the previous total those arcs
going from $v$ to $C'(i)$ and from $v$ to $P(i)$.

Let $C = \max_i[|C_i|]$. Then the size of the dynamic programming
table constructed by this dynamic programming algorithm is at most
$n2^{|C|}$, and the running time of the algorithm is polynomial in
the size of the table. Hence to prove Theorem~\ref{thm:main} it
remains to prove the following proposition.

\begin{proposition}
\label{pro:CandidateSize} For a choice of $t$ and $d$ as above,
for every $i$, the size of the candidate set $C(i)$ is at most
$52\sqrt{k}$.
\end{proposition}

\begin{proof}
The candidate set $C(i)$ contains all of $B$, which by our choice
of $t = 3\sqrt{k}$ and Proposition~\ref{pro:badsmall} contains at
most $3\sqrt{k}$ vertices. In addition it contains those vertices
not in $D_\pi$ whose indegree is between $i - d$ and $i + d$.
There are at most $4d+1$ such vertices (because any such vertex
has to be in a location between $i - 2d$ and $i + 2d$ in $\pi$),
and by our choice of $d = 12\sqrt{k}$ this contributes at most
$48\sqrt{k} + 1$ additional vertices to $C(i)$.
\end{proof}

In summary, the algorithm presented above runs in time
$n^{O(1)}2^{O(\sqrt{k})}$ and finds a feedback arc set of size $k$
in an $n$-vertex tournament, if the tournament has such a feedback
arc set. This proves Theorem~\ref{thm:main}.

\section{Conclusions}

Are there algorithms for $k$-FAST with running times that are
substantially better than $n^{O(1)}2^{O(\sqrt{k})}$? Specifically,
can we extend the range of values of $k$ for which the running
time is polynomial beyond $k = O((\log n)^2)$? If yes, then this
will imply that FAST can be solved in time $2^{o(n)}$ (details
omitted). The NP-hardness results of~\cite{ACN,Alon} do not show
that a running time of $2^{o(n)}$ for FAST is unlikely (e.g., they
do not show that this will imply a similar running time for SAT).
But still, solving FAST in time $2^{o(n)}$ seems to require
substantially new techniques, and hence the author does not
anticipate major improvements over the bounds in the current paper
in the near future.

If major improvements are not to be expected, what about minor
improvements? Here much can be done. Kernelization techniques
(such as in~\cite{DGHNT} and~\cite{ALS}) can offer improvements
that are significant when $k$ is small. For larger values of $k$,
improvements can come from optimizing the values of parameters
(such as of $d$ and $t$) so as to minimize the value of the hidden
constant in the exponent of the $2^{O(\sqrt{k})}$ term. Moreover,
some minor modifications to the algorithm can lead to further
improvements. We have not attempted any of these optimizations in
the current paper, because this surely deserves a separate
publication, to be titled {\em Fastest faster FAST}.

\subsection*{Acknowledgements}

The trigger to this work was a talk given by Noga Alon at the IPAM
Workshop on Probabilistic Techniques and Applications, October
2009.

\end{document}